\documentclass{article}
\RequirePackage{textcomp}
\RequirePackage{latexsym}
\RequirePackage{amssymb}
\RequirePackage{amstext}
\RequirePackage{amsmath}
\RequirePackage{amssymb}
\RequirePackage{amsthm}
\usepackage{fullpage}
\usepackage{hyperref}
\usepackage{mdwlist}
\usepackage{float}
\usepackage{appendix}

\theoremstyle{plain}
\newtheorem{theorem}{Theorem}[section]
\newtheorem{lemma}[theorem]{Lemma}

\newtheorem{fact}[theorem]{Fact}

\theoremstyle{definition}

\def\ep{\varepsilon}

\def\Var{\text{Var}}

\def\FF{\mathcal{F}}

\def\fhat{\hat{f}}
\def\dhat{\hat{d}}

\def\that{\hat{t}}

\def\m{\!-\!}
\def\eq{\!=\!}
\def\cnd{\text{\textsc{condition}}}
\def\Cnd{\text{\textsc{Condition}}}
\def\upd{\text{\textsc{update}}}
\def\Upd{\text{\textsc{Update}}}

\def\msg{\text{\textsc{message}}}
\def\Msg{\text{\textsc{Message}}}

\title{Variability in data streams}

\author{
David Felber
\thanks{University of California at Los Angeles. \texttt{dvfelber@cs.ucla.edu}.}
\and
Rafail Ostrovsky
\thanks{University of California at Los Angeles. \texttt{rafail@cs.ucla.edu}.}
}
\date{}

\begin{document}

\maketitle

\begin{abstract}

  We consider the problem of tracking with small relative error an integer
  function $f(n)$ defined by a distributed update stream $f'(n)$. Existing
  streaming algorithms with worst-case guarantees for this problem assume $f(n)$
  to be monotone; there are very large lower bounds on the space requirements
  for summarizing a distributed non-monotonic stream, often linear in the size
  $n$ of the stream.

  Input streams that give rise to large space requirements are highly variable,
  making relatively large jumps from one timestep to the next. However, streams
  often vary slowly in practice. What has heretofore been lacking is a framework
  for non-monotonic streams that admits algorithms whose worst-case performance
  is as good as existing algorithms for monotone streams and degrades gracefully
  for non-monotonic streams as those streams vary more quickly.

  In this paper we propose such a framework. We introduce a new stream
  parameter, the ``variability'' $v$, deriving its definition in a way that
  shows it to be a natural parameter to consider for non-monotonic streams. It
  is also a useful parameter. From a theoretical perspective, we can adapt
  existing algorithms for monotone streams to work for non-monotonic streams,
  with only minor modifications, in such a way that they reduce to the monotone
  case when the stream happens to be monotone, and in such a way that we can
  refine the worst-case communication bounds from $\Theta(n)$ to $\tilde{O}(v)$.
  From a practical perspective, we demonstrate that $v$ can be small in practice
  by proving that $v$ is $O(\log f(n))$ for monotone streams and $o(n)$ for
  streams that are ``nearly'' monotone or that are generated by random walks. We
  expect $v$ to be $o(n)$ for many other interesting input classes as well.

\end{abstract}

\section{Introduction}
\label{sec:introduction}

In the distributed monitoring model, there is a single central monitor and
several ($k$) observers. The observers receive data and communicate with the
monitor, and the goal is to maintain at the monitor a summary of the data
received at the observers while minimizing the communication between them.

This model was introduced by Cormode, Muthukrishnan, and Yi \cite{CMY2008}
\cite{CMY2011} with the motivating application of minimizing radio energy usage
in sensor networks, but can be applied to other distributed applications like
determining network traffic patterns. Since the monitor can retain all messages
received, algorithms in the model can be used to answer historical queries too,
making the model useful for auditing changes to and verifying the integrity of
time-varying datasets.

The distributed monitoring model has also yielded several theoretical results.
These include algorithms and lower bounds for tracking total count
\cite{CMY2008} \cite{CMY2011} \cite{LRV2011} \cite{LRV2012},
%
%
frequency moments \cite{CMY2008} \cite{CMY2011} \cite{WZ2011} \cite{WZ2012},
item frequencies \cite{HYZ2012} \cite{WZ2011} \cite{WZ2012} \cite{YZ2009}
\cite{YZ2013}, quantiles \cite{HYZ2012} \cite{WZ2011} \cite{WZ2012}
\cite{YZ2009} \cite{YZ2013}, and entropy \cite{ABC2009} \cite{WZ2011}
\cite{WZ2012} to small relative error.

However, nearly all of the upper bounds assumed that data is only inserted and
never deleted. This is unfortunate because in the standard turnstile streaming
model, all of these problems have similar algorithms that permit both insertions
and deletions. In general, this unfortunate situation is unavoidable; existing
lower bounds for the distributed model \cite{ABC2009} demonstrate that it is not
possible to track even the total item count in small space when data is
permitted to be deleted.

That said, when restrictions are placed on the types of allowable input, the
lower bounds evaporate, and very nice upper bounds exist. Tao, Yi, Sheng, Pei,
and Li \cite{TYSPL2010} developed algorithms for the problem of summarizing the
order statistics history of a dataset $D$ over an insertion/deletion stream of
size $n$, which has an $\Omega(n)$-bit lower bound in general; however, they
performed an interesting analysis that yielded online and offline upper bounds
proportional to $\sum_{t=1}^n 1 / |D(t)|$, with a nearly matching lower bound. A
year or two later, Liu, Radunovi{\'c}, and Vojnovi{\'c} \cite{LRV2011}
\cite{LRV2012} considered the problem of tracking $|D|$ under random inputs; for
general inputs, there is an $\Omega(n)$-bit lower bound, but Liu et. al.
obtained (among other results) expected communication costs proportional to
$\sqrt{n} \log n$ when the insertion/deletion pattern is the result of fair coin
flips.

In fact, the pessimistic lower bounds for the general case can occur only when
the input stream is such that the quantity being tracked is forced to vary
quickly. In the problems considered by Tao et. al. and Liu et. al., this occurs
when $|D|$ is usually small. These two groups avoid this problem in two
different ways: Tao et. al. provide an analysis that yields a worst-case upper
bound that is small when $|D|$ is usually large, and Liu et. al. consider input
classes for which $|D|$ is usually large in expectation.

\paragraph{Our contributions}
In this paper we propose a framework that extends the analysis of Tao et. al. to
the distributed monitoring model and that permits worst-case analysis that can
be specialized for random input classes considered by Liu et. al. In so doing,
we explain the intuition behind the factor of $\sum_{t=1}^n 1 / |D(t)|$ in the
bounds of Tao et. al. and how we can separate the different sources of
randomness that appear in the algorithms of Liu et. al. to obtain worst-case
bounds for the random input classes we also consider.


In the next section we derive a stream parameter, the variability $v$. We prove
that $v$ is $O(\log f(n))$ for monotone streams and $o(n)$ for streams that are
``nearly'' monotone or that are generated by random walks, and find that the
bounds of Tao et. al. and Liu et. al. are stated nicely in terms of $v$.
In section \ref{sec:upperbounds} we combine ideas from the upper bounds of Tao
et. al. \cite{TYSPL2010} with the existing distributed counting algorithms of
Cormode et. al. \cite{CMY2008} \cite{CMY2011} and Huang, Yi, and Zhang
\cite{HYZ2012} to obtain upper bounds for distributed counting that are
proportional to $v$.
In section \ref{sec:lowerbounds} we show that our dependence on $v$ is
essentially necessary by developing deterministic and randomized
space+communication lower bounds that hold even when $v$ is small.
We round out the piece in section \ref{sec:framework} with a discussion of the
suitability of variability as a general framework, in which we extend the ideas
of section \ref{sec:upperbounds} to the problems of distributed tracking of item
frequencies and of tracking general aggregates when $k = 1$.

But before we jump into the derivation of variability, we define our problem
formally and abstract away unessential details.

\paragraph{Problem definition}
The problem is that of tracking at the coordinator an integer function $f(n)$
defined by an update stream $f'(n)$ that arrives online at the sites. Time
occurs in discrete steps; to be definite, the first timestep is $1$, and we
define $f(0) = 0$ unless stated otherwise. At each new current time $n$ the
value $f'(n) = f(n) - f(n \m 1)$ appears at a single site $i(n)$.

There is an error parameter $\ep$ that is specified at the start. The
requirement is that, after each timestep $n$, the coordinator must have an
estimate $\fhat(n)$ for $f(n)$ that is usually good. In particular, for
deterministic algorithms we require that $\forall n,\; |f(n) \m \fhat(n)| \le
\ep f(n)$, and for randomized algorithms we require that $\forall n,\; P(|f(n)
\m \fhat(n)| \le \ep f(n)) \:\ge\: 2/3$.

\section{Variability}
\label{sec:variability}

In the original distributed monitoring paper \cite{CMY2008}, Cormode et. al.
define a general thresholded problem $(k, f, \tau, \ep)$. A dataset $D$ arrives
as a distributed stream across $k$ sites. At any given point in time, the
coordinator should be able to determine whether $f(D) \ge \tau$ or $f(D) \le (1
\m \ep) \tau$.

In continuous tracking problems, there is no single threshold, and so $f(n)$ is
tracked to within an additive $\ep \tau(n)$, where $\tau(n)$ also changes with
the dataset $D(n)$. Since $\tau$ is now a function, it needs to be defined; the
usual choice is $f$ itself, except for tracking item frequencies and order
statistics, for which (following the standard streaming model) $\tau$ is chosen
to be $|D|$. That is, the continuous monitoring problem $(k, f, \ep)$ is, at all
times $n$ maintain at the coordinator an estimate $\fhat(n)$ of $f(n)$ so that
$|f(n) \m \fhat(n)| \le \ep f(n)$.

The motivation for the way we define variability is seen more easily if we first
look at the situation as though item arrivals and communication occur
continuously. That is, over $n = [0.1, 0.2]$ we receive the second tenth of the
first item, for example. At any time $t$ at which $f$ changes by $\pm \ep f$, we
would need to communicate at least one message to keep the coordinator in sync;
so if $f$ changes by $f'(t) \, dt$ then we should communicate $|\frac{f'(t) \,
  dt}{\ep f(t \!+\! dt)}|$ messages.

With discrete arrivals, $dt = 1$, and we define $f'(t) = f(t) - f(t \m 1)$.
Otherwise, the idea remains the same, so we would expect the total number of
messages to look like $\sum_{t = 1}^n |\frac{f'(t)}{\ep f(t)}|$, where here
$f'(t) = f(t) - f(t \m 1)$. In sections \ref{sec:upperbounds} and
\ref{sec:lowerbounds} we find that, modulo the number $k$ of sites and constant
factors, this is indeed the case.

Being a parameter of the problem rather than of the stream, we can move the
$1/\ep$ factor out of our definition of variability and bring it back in along
with the appropriate functions of $k$ when we state upper and lower bounds for
our problem. This permits us to treat the stream parameter $v$ independently of
the problem. We also need to handle the case $f = 0$ specially, which we can do
by communicating at each timestep that case occurs. This means we can define
$|\frac{f'(t)}{f(t)}| = 1$ when $f(t) = 0$.

Taking all of these considerations into account, we define the
\emph{$f$-variability} of a stream to be $v(n) = \sum_{t = 1}^n \min \{1,
|\frac{f'(t)}{f(t)}|\}$. We also write $v'(t) = \min \{1,
|\frac{f'(t)}{f(t)}|\}$ to be the increase in variability at time $t$. We say
``variability'' for $f$-variability in the remainder of this paper.

From a practical perspective, we believe low variability streams to be common.
In many database applications the database is interesting primarily because it
tends to grow more than it shrinks, so it is common for the size of the dataset
to have low variability; as more items are inserted, the rate of change of $|D|$
shrinks relative to itself, and about as many deletions as insertions would be
required to keep the ratio constant. In the following subsection, we prove that
monotone and nearly monotone functions have low variability and that random
walks have low variability in expectation, lending evidence to our belief.

From a theoretical perspective, variability is a way to analyze algorithms for
$\ep$ relative error in the face of non-monotonicity and generate provable
worst-case bounds that degrade gracefully as our assumptions about the input
become increasingly pessimistic. For our counting problem, it allows us to adapt
the existing distributed counting algorithms of Cormode et. al. \cite{CMY2008}
\cite{CMY2011} and Huang et. al. \cite{HYZ2012} with only minor modifications,
and the resulting analyses show that the dependence on $k$ and $\ep$ remains
unchanged.

\subsection{Interesting cases with small variability}
\label{sec:variability-smallvar}

We start with functions that are nearly monotone in the sense that they are
eventually mostly nondecreasing. We make this precise in the theorem statement.

\begin{theorem}
  \label{thm:smallvar-cases-nearlymonotone}
  Let $f^-(n) = \sum_{t : f'(t) < 0} |f'(t)|$ and $f^+(n) = \sum_{t : f'(t) > 0}
  f'(t)$. If there is a monotone nondecreasing function $\beta(t) \ge 1$ and a
  constant $t_0$ such that for all $n \ge t_0$ we have $f^-(n) \le \beta(n)
  f(n)$, then the variability $\sum_{t=1}^n |f'(t)/f(t)|$ is $O(\beta(n)
  \log(\beta(n) f(n)))$.
\end{theorem}

The proof, which we defer to appendix
\ref{sec:appendix-variability-nearlymonotone}, partitions time into intervals
over which $f^+(t)$ doubles and shows the variability in each interval to be
$O(\beta(n))$. When $f(n)$ is strictly monotone, $\beta(n) = 1$ suffices, and
the theorem reduces to the result claimed in the abstract. As we will see in
section \ref{sec:upperbounds}, our upper bounds will simplify in the monotone
case to those of Cormode et. al. \cite{CMY2008} \cite{CMY2011} and Huang et. al.
\cite{HYZ2012}.

Next, we compute the variability for two random input classes considered by Liu
et. al. \cite{LRV2011} \cite{LRV2012}. This will permit us to decouple the
randomness of their algorithms from the randomness of their inputs. This means,
for example, that even our deterministic algorithm of section
\ref{sec:upperbounds} has $o(n)$ cost in expectation for these input classes.
The first random input class we consider is the symmetric random walk.

\begin{theorem}
  \label{thm:smallvar-cases-unbiased}
  If $f'(t)$ is a sequence of i.i.d. $\pm 1$ coin flips then the expected
  variability $E(v(n)) = O(\sqrt{n} \log n)$.
\end{theorem}
\begin{proof}
  The update sequence defines a random walk for $f(t)$, and the expected
  variability is
  \begin{equation*}
    \sum_{t=1}^n P(f(t) \!=\! 0) \;+\; \sum_{t=1}^n \sum_{s=1}^t 2 P(f(t) \!=\! s) / s
  \end{equation*}

  We use the following fact, mentioned and justified in Liu et. al.
  \cite{LRV2011}:
  \begin{fact}
    \label{thm:smallvar-cases-lemma}
    For any $t \ge 1$ and $s \in [-t, t]$ we have $P(f(t) \!=\! s) \le c_1 /
    \sqrt{t}$, where $c_1$ is some constant.
  \end{fact}

  Together, these show the expected cost to be at most
  \begin{equation*}
    c_1 \sum_{t=1}^n (1 + 2 H_t) / \sqrt{t}  \;\le\;
    c_2 \log(n) \sum_{t=1}^n 1 / \sqrt{t}  \;\le\;
    c_3 \log(n) \sqrt{n}
  \end{equation*}
  since $(1 + 2 H_n) \le \frac{c_2}{c_1} \log(n)$ and $\sum_{t=1}^n 1 / \sqrt{t}
  \le \frac{c_3}{2 c_2} \int_1^n 1 / \sqrt{t} \; dt$.
\end{proof}

The second random input class we consider is i.i.d. increments with a common
drift rate of $\mu > 0$. The case $\mu < 0$ is symmetric. We assume that $\mu$
is constant with respect to $n$. The proof is a simple application of Chernoff
bounds and is deferred to appendix \ref{sec:appendix-variability-biased}.

\begin{theorem}
  \label{thm:smallvar-cases-biased}
  If $f'(t)$ is a sequence of i.i.d. $\pm 1$ random variables with $P(f'(t)
  \!=\! 1) = (1+\mu)/2$ then $E(v(n)) = O(\frac{\log n}{\mu})$.
\end{theorem}

\paragraph{Remarks}
We can restate the results of Liu et. al. \cite{LRV2011} \cite{LRV2012} and Tao
et. al. \cite{TYSPL2010} in terms of variability. For unbiased coin flips, Liu
et. al. obtain an algorithm that uses $O(\frac{\sqrt{k}}{\ep} \sqrt{n} \log n)$
messages (of size $O(\log n)$ bits each) in expectation, and for biased coin
flips with constant $\mu$, an algorithm that uses $O(\frac{\sqrt{k}}{\ep}
\frac{1}{|\mu|} (\log n)^{1+c})$ messages in expectation. If we rewrite these
bounds in terms of expected variability, they become $O(\frac{\sqrt{k}}{\ep}
E(v(n)))$ and $O(\frac{\sqrt{k}}{\ep} (\log n)^c E(v(n)))$, respectively. In the
next section, we obtain (when $k = O(1/\ep^2)$) a randomized bound of
$O(\frac{\sqrt{k}}{\ep} v(n))$. In marked contrast to the bounds of Liu et. al.,
our bound is a worst-case lower bound that is a function of $v(n)$; if the input
happens to be generated by fair coin flips, then our expected cost happens to be
$O(\frac{\sqrt{k}}{\ep} \sqrt{n} \log n)$.

The results of Tao et. al. are for a different problem, but they can still be
stated nicely in terms of the $|D|$-variability $v(n)$: for the problem of
tracking the historical record of order statistics, they obtain a lower bound of
$\Omega(\frac{1}{\ep} v(n))$ and offline and online upper bounds of
$O((\frac{1}{\ep} \log^2 \frac{1}{\ep}) v(n))$ and $O(\frac{1}{\ep^2} v(n))$,
respectively. We adapt ideas from both their upper and lower bounds in sections
\ref{sec:upperbounds} and \ref{sec:lowerbounds}.

\section{Upper bounds}
\label{sec:upperbounds}

In this section we develop deterministic and randomized algorithms for
maintaining at the coordinator an estimate $\fhat(n)$ for $f(n)$ that is usually
good. In particular, for deterministic algorithms we require that $\forall n,\;
|f(n) \m \fhat(n)| \le \ep f(n)$, and for randomized algorithms that $\forall
n,\; P(|f(n) \m \fhat(n)| \le \ep f(n)) \:\ge\: 2/3$. We obtain deterministic
and randomized upper bounds of $O(\frac{k}{\ep} v(n))$ and $O((k +
\frac{\sqrt{k}}{\ep}) v(n))$ messages, respectively. For comparison, the
analogous algorithms of Cormode et. al. \cite{CMY2008} \cite{CMY2011} and Huang
et. al. \cite{HYZ2012} use $O(\frac{k}{\ep} \log n)$ and $O((k +
\frac{\sqrt{k}}{\ep}) \log n)$ messages, respectively.

For our upper bounds we assume that $f'(n) = \pm 1$ always. If $|f'(n)| > 1$ we
could simulate it with $|f'(n)|$ arrivals of $\pm 1$ updates with $O(\log \max
f'(n))$ overhead, as shown in appendix \ref{sec:appendix-upperbound-singleitem}.

\subsection{Partitioning time}
\label{sec:upperbounds-block}

We use an idea from Tao et. al. \cite{TYSPL2010} to first divide time into
manageable blocks. At the end of each block we know the values $n$ and $f(n)$
exactly. Within each block, we know these values only approximately. The
division into blocks is deterministic and the same for both our deterministic
and randomized algorithms. Our division ensures that the change in $v(n)$ over
each block is at least $1/5$, which simplifies our analysis.

\begin{itemize*}
\item The coordinator requests the sites' values $c_i$ and $f_i$ at times $n_0
  \eq 0, n_1, n_2, \ldots$ and then broadcasts a value $r$. These values will be
  defined momentarily.
\item Each site $i$ maintains a variable $c_i$ that counts the number of stream
  updates $f'(n)$ it received since the last time it sent $c_i$ to the
  coordinator. It also maintains $f_i$ that counts the change in $f$ it received
  since the last broadcast $n_j$. Whenever $c_i = \lceil 2^{r-1} \rceil$, site
  $i$ sends $c_i$ to the coordinator. This is in addition to replying to
  requests from the coordinator.
\item The coordinator maintains a variable $\that$. After broadcasting $r$,
  $\that$ is reset to zero. Whenever site $i$ sends $c_i$, the coordinator
  updates $\that = \that + c_i$.
\item The coordinator also maintains variables $\fhat$, $j$, and $t_j$. At the
  first time $n_j > n_{j-1}$ at which $\that \ge t_j$, the coordinator requests
  the $c_i$ and $f_i$ values, updates $\fhat$ and $r$, sets $t_{j+1} = \lceil
  2^{r-1} \rceil k$, broadcasts $r$, and increments $j$.
\item When $r$ is updated at the end of time $n_j$, it is set to $r$ if $2^r 2 k
  \le |f(n_j)| < 2^r 4 k$ and zero if $|f(n_j)| < 4 k$.
\end{itemize*}

Thus we divide time into blocks $B_0, B_1, \ldots$, where $B_j = [n_j+1,
  n_{j+1}]$. Algebra tells us some facts:
\begin{itemize*}
\item $\lceil 2^{r-1} \rceil k \;\le\; n_{j+1} - n_j \;\le\; 2^r k$.
\item If $r = 0$ then $|f(n) - f(n_j)| \le k$ and $|f(n)| \le 5 k$ for all $n$
  in $B_j$.
\item If $r \ge 1$ then $|f(n) - f(n_j)| \le 2^r k$ and $2^r k \le |f(n)| \le
  2^r 5 k$ for all $n$ in $B_j$.
\end{itemize*}

The total number of messages sent in block $B_j$ is at most $5 k$: we have at
most $2 k$ updates from sites, $k$ in requests from the coordinator, $k$ replies
from each site, and $k$ broadcast at $n_{j+1}$.

The change in variability $v_j = v(n_{j+1}) - V(n_j)$ over block $B_j$ is
\begin{align*}
  v(n_{j+1}) - v(n_j)
  \;&=\; \sum_{t = n_j+1}^{n_{j+1}} \frac{1}{\min\{1, |f(t)|\}}
  \;\ge\;
  \left\{
  \begin{matrix}
    k / 5 k  &  \text{ if $r = 0$}   \\
    2^r k / 2^r 5 k  &  \text{ if $r \ge 1$}
  \end{matrix}
  \right\}
  \;\ge\; 1/5
\end{align*}
And therefore the total number of messages (all $O(\log n)$ bits in size) is
bounded by $25 k v + 3 k$.

\subsection{Estimation inside blocks}

What remains is to estimate $f(n)$ within a given block. Since we have
partitioned time into constant-variability blocks, we can use the algorithms of
Cormode et. al. \cite{CMY2008} \cite{CMY2011} and Huang et. al. \cite{HYZ2012}
almost directly. Both of our algorithms use the following template, changing
only \cnd, \msg, and \upd:

\begin{itemize*}
\item Site $i$ maintains a variable $d_i$ that tracks the drift at site $i$,
  defined as the sum of $f'(n)$ updates received at site $i$ during the block.
  That is, $f(n) - f(n_j) = \sum_i d_i$.
\item Site $i$ also maintains a variable $\delta_i$ that tracks the change in
  $d_i$ since the last time site $i$ sent a \msg. $\delta_i$ is initially zero.
\item The coordinator maintains an estimate $\dhat_i$ for each value $d_i$.
  These are initially zero. It also defines two estimates based on these
  $\dhat_i$:
  \begin{itemize*}
  \item For the global drift: $\dhat = \sum_i \dhat_i$.
  \item For $f(n)$: $\fhat(n) = f(n_j) + \dhat(n)$.
  \end{itemize*}
\item When site $i$ receives stream update $f'(n)$, it updates $d_i$. It then
  checks its \cnd. If true, it sends a \msg{} to the coordinator and resets
  $\delta_i = 0$.
\item When the coordinator receives a \msg{} from a site $i$ it \upd{}s its
  estimates.
\end{itemize*}

\subsection{The deterministic algorithm}

Our method guarantees that at all times $n$ we have $|f(n) - \fhat(n)| \le \ep
|f(n)|$. It uses $O(kv/\ep)$ messages in total.

\begin{itemize*}
\item \Cnd: true if $|\delta_i| = 1$ and $r = 0$, or if $|\delta_i| \ge \ep
  2^r$. Otherwise, false.
\item \Msg: the new value of $d_i$.
\item \Upd: set $\dhat_i = d_i$.
\end{itemize*}

Let $\delta = \sum_i \delta_i$ be the error with which $\dhat$ estimates $d = \sum_i
d_i$. The error in $\fhat$ is
\begin{align*}
  |f(n) - \fhat(n)|
  \;&=\; |(f(n_j) + d(n)) \,-\, (\fhat(n_j) + d(n) + \delta(n))|
  \;=\; |\delta(n)|
\end{align*}
When $r \ge 1$ we have $|B_j| \le 2^r k$, and we always have that $\delta \le
|B_j|$. Since we constrain $\delta_i < \ep 2^r$ at the end of each timestep, we
have $|f(n) - \fhat(n)| \;<\; \ep 2^r k \;\le\; \ep |f(n)|$.

We also use at most $2 k / \ep$ messages for the block. If $r = 0$ then the
number of messages is at most $k$. If $r \ge 1$, then since a site must receive
$\ep 2^r$ new stream updates to send a new message, and since there are at most
$2^r k$ stream updates in the block, there are at most $k/\ep$ messages.

In each block the change in $v$ is at least $1/5$, so the total number of
messages is at most $5 k v / \ep$.

\subsection{The randomized algorithm}
\label{sec:upperbounds-rand}

Our method uses $O(\sqrt{k}v/\ep)$ messages (plus the time partitioning) and
guarantees that at all times $n$ we have $P(|f(n) - \fhat(n)| > \ep |f(n)|)
\;<\; 1/3$.

The idea is to estimate the sums $d_i^+$ and $d_i^-$ separately. The estimators
for those values are independent and monotone, so we can use the method of Huang
et. al. \cite{HYZ2012} to estimate the two and then combine them.

Specifically, the coordinator and each site run two independent copies $A^+$ and
$A^-$ of the algorithm. Whenever $f'(n) = +1$ arrives at site $i$, a $+1$ is fed
into algorithm $A^+$ at site $i$. Whenever $f'(n) = -1$ arrives at site $i$, a
$+1$ is fed into algorithm $A^-$ at site $i$. So the drifts $d_i^+$ and $d_i^-$
at every site will always be nonnegative. At the coordinator, the estimates
$\dhat_i^{\pm}$ and $\dhat^{\pm}$ are tracked independently also. However, the
coordinator also defines $\dhat = \dhat^+ - \dhat^-$ and $\fhat(n) = f(n_j) +
\dhat(n)$. The definitions for algorithm $A^{\pm}$ are

\begin{itemize*}
\item \Cnd: true with probability $p = \min\{1, 3 / \ep 2^r k^{1/2}\}$.
\item \Msg: the new value of $d_i^{\pm}$.
\item \Upd: set $\dhat_i^{\pm} = d_i^{\pm} - 1 + 1/p$.
\end{itemize*}

The following fact \ref{fac:upperbounds-rand-proof-hyz} is lemma 2.1 of Huang
et. al. \cite{HYZ2012}. Our algorithm effectively divides the stream $f'(B_j)$
into two streams $|f'(B_j^{\pm})|$. Since these streams consist of $+1$
increments only we run the algorithm of Huang et. al. separately on each of
them. At any time $n$, stream $|f'(B_j^{\pm})|$ has seen $d_i^{\pm}(n)$
increments at site $i$, and lemma 2.1 of Huang et. al. guarantees that the
estimates $\dhat_i^{\pm}(n)$ for the counts $d_i^{\pm}(n)$ are good.
\begin{fact}
  \label{fac:upperbounds-rand-proof-hyz}
  $E(\dhat_i^{\pm}) = d_i^{\pm}$ and $\Var(\dhat_i^{\pm}) \le 1/p^2$.
\end{fact}
This means that $E(\dhat^{\pm}) = \sum_i E(\dhat_i^{\pm}) = \sum_i d_i^{\pm}$,
and therefore that $E(\dhat) = \sum_i E(d_i^+ - d_i^-) = \sum_i d_i$. Since the
estimators $\dhat_i^{\pm}$ are independent, the variance of the global drift is
at most $2 k / p^2$. By Chebyshev's inequality,
%
%
\begin{align*}
  P(|\delta(n)| > \ep 2^r k)
  \;&\le\; \frac{2 k / p^2}{(\ep 2^r k)^2}
  \;<\; 1/3
\end{align*}
Further, the expected cost of block $B_j$ is at most
$p |B_j| \le (3 / \ep 2^r k^{1/2}) (2^r 2 k) \le 30 k^{1/2} v_j / \ep$.

\section{Lower bounds}
\label{sec:lowerbounds}

In this section we show that the dependence on $v$ is essentially necessary by
developing deterministic and randomized lower bounds on space+communication that
hold even when $v$ is small. Admittedly, this is not as pleasing as a pure
communication lower bound would be. On the other hand, a distributed monitoring
algorithm with high space complexity would be impractical for monitoring sensor
data, network traffic patterns, and other applications of the model. Note that
in terms of space+communication, our deterministic lower bound is tight up to
factors of $k$, and our randomized lower bound is within a factor of $\log(n)$
of that.

For these lower bounds we use a slightly different problem. We call this problem
the tracing problem. The streaming model for the tracing problem is the standard
turnstile streaming model with updates $f'(n)$ arriving online. The problem is
to maintain in small space a summary of the sequence $f$ so that, at any current
time $n$, if we are given an earlier time $t$ as a query, we can return an
estimate $\fhat(t)$ so that $P(|f(t) \m \fhat(t)| \le \ep f(t))$ is large (one
in the deterministic case, $2/3$ in the randomized case). We call this the
tracing problem because our summary traces $f$ through time, so that we can look
up earlier values.

In appendix \ref{sec:appendix-lowerbounds-lem0} we show that a space lower bound
for the tracing problem implies a space+communication lower bound for the
distributed tracking problem. Here, we develop deterministic and randomized
space lower bounds for the tracing problem.

\subsection{The deterministic bound}

The deterministic lower bound that follows is similar in spirit to the lower
bound of Tao et. al. \cite{TYSPL2010}. It uses a simple information-theoretic
argument.

\begin{theorem}
  \label{thm:lowerbounds-det}
  Let $\ep = 1 / m$ for some integer $m \ge 2$, let $n \ge 2 m$, let $c < 1$
  constant, and let $r \le n^c$ and even. If a deterministic summary $S(f)$
  guarantees, even only for sequences for which $v(n) = \frac{6 m + 9}{2 m + 6}
  \ep r$, that $|f(t) - \fhat(t)| \le \ep f(t)$ for all $t \le n$, then that
  summary must use $\Omega(\frac{\log n}{\ep} v(n))$ bits of space.
\end{theorem}

The full proof appears in appendix \ref{sec:appendix-lowerbounds-det}. At a high
level, the sequences in the family take only values $m$ or $m + 3$, and each
sequence is defined by $r$ of the $n$ timesteps. If the new timestep $t$ is one
of the $r$ chosen for our sequence, then we flip from $m$ to $m + 3$ or
vice-versa. All of these sequences are unique and there are $2^{\Omega(r \log
  n)}$ of them.

\subsection{The randomized bound}

We use a construction similar to the one in our deterministic lower bound to
produce a randomized lower bound. In order to make the analysis simple we forego
a single variability value for all sequences in our constructed family, but
still maintain that they all have low variability. $C$ is a universal constant
to be defined later.

\begin{theorem}
  \label{thm:lowerbounds-rand}
  Choose $\ep \le 1/2$, $v \ge 32400 \ep \ln C$, and $n > 3 v / \ep$. If a
  summary $S(f)$ guarantees, even only for sequences for which $v(n) \le v$,
  that $P(|f(t) - \fhat(t)| \le \ep f(t)) \;\ge\; 99/100$ for all $t \le n$,
  then that summary must use $\Omega(v/\ep)$ bits of space.
\end{theorem}

We prove this theorem in two lemmas. In the first lemma, we reduce the claim to
a claim about the existence of a hard family of sequences. In the second lemma
we show the existence of such a family.

First a couple of definitions. For any two sequences $f$ and $g$ define the
\emph{number of overlaps between $f$ and $g$} to be the number of positions $1
\le t \le n$ for which $|f(t) - g(t)| \le \ep \max \{f(t), g(t)\}$. Say that $f$
and $g$ \emph{match} if they have at least $\frac{6}{10} n$ overlaps.

\begin{lemma}
  \label{lem:lowerbounds-rand-lem1}
  Let $\FF$ be a family of sequences of length $n$ and variabilities $\le v$
  such that no two sequences in $\FF$ match. If a summary $S(f)$ guarantees for
  all $f$ in $\FF$ that $P(|f(t) - \fhat(t)| \le \ep f(t)) \;\ge\; 99/100$ for
  all $t \le n$, then that summary must use $\Omega(\log |\FF|)$ bits of space.
\end{lemma}

The full proof appears in appendix \ref{sec:appendix-lowerbounds-rand-lem1}. At
a high level, if $S(f)$ is the summary for a sequence $f$, we can use it to
generate an approximation $\fhat$ that at least $90\%$ of the time overlaps with
$f$ in at least $\frac{9}{10} n$ positions. Since no two sequences in $\FF$
overlap in more than $\frac{6}{10} n$ positions, at least $90\%$ of the time we
can determine $f$ given $\fhat$. We then solve the one-way $\text{Index}_N$
problem by deterministically generating $\FF$ and sending a summary $S(f(x))$,
where $x$ is Alice's input string of size $N = \log_2 |\FF|$, and $f(x)$ is the
$x$th sequence in $\FF$.

\newpage
\begin{lemma}
  \label{lem:lowerbounds-rand-lem2}
  For all $\ep \le 1/2$, $v \ge 32400 \ep \ln C$, and $n > 3 v / \ep$, there is
  a family $\FF$ of size $e^{\Omega(v/\ep)}$ of sequences of size $n$ such that:
  \begin{enumerate*}
  \item no two sequences match, and
  \item every sequence has variability at most $v$.
  \end{enumerate*}
\end{lemma}

The full proof appears in appendix \ref{sec:appendix-lowerbounds-rand-lem2}. At
a high level, sequences again switch between $m = 1/\ep$ and $m \!+\! 3$, except
that these switches are chosen independently. We model the overlap with a Markov
chain; the overlap between any two sequences is the sum over times $t$ of a
function $y$ applied to the states of a chain modeling their interaction. We
then apply a result of Chung, Lam, Liu, and Mitzenmacher \cite{CLLM2012} to show
that the probability that any two sequences match is low. Lastly, we show that
not too many sequences have variability more than $v$, by proving that they
usually don't switch between $m$ and $m \!+\! 3$ many times.

\section{Variability as a framework}
\label{sec:framework}

In section \ref{sec:variability} we proposed the $f$-variability $\sum_{t=1}^n
\min \{1, |\frac{f'(t)}{f(t)}|\}$ as a way to analyze algorithms for the
continuous monitoring problem $(k, f, \ep)$ over general update streams.
However, our discussion so far has focused on distributed counting. In this
final section we revisit the suitability of our definition by mentioning
extensions to tracking other functions of a dataset defined by a distributed
update stream. We include fuller discussions of these extensions in the
appendices.

\subsection{Tracking item frequencies}
\label{sec:framework-itemfreq}

We can extend our deterministic algorithm of section \ref{sec:upperbounds} to
the problem of tracking item frequencies, in a manner similar to that in which
Yi and Zhang \cite{YZ2009} \cite{YZ2013} extend the ideas of Cormode et. al.
\cite{CMY2008} to this problem. The definition of this problem, the required
changes to our algorithm of section \ref{sec:upperbounds} needed to solve this
problem, and a discussion of the difficulties in finding a randomized algorithm,
are discussed in appendix \ref{sec:appendix-framework-itemfreq}.

\subsection{Aggregate functions with one site}
\label{sec:framework-gf}

In this subsection we consider general single-integer-valued functions $f$ of a
dataset. When there is a single site, the site always knows the exact value of
$f(n)$, and the only issue is updating the coordinator to have an approximation
$\fhat(n)$ so that $|f(n) - \fhat(n)| \le \ep f(n)$ for all $n$. We can show
that this problem of tracking $f$ to $\ep$ relative error when $k = 1$ has an
$O(\frac{1}{\ep} v(n))$-word upper bound, where here $v(n)$ is the
$f$-variability. The algorithm is: whenever $|f - \fhat| > \ep f$, send $f$ to
the coordinator. The proof is a simple potential argument and is deferred to
appendix \ref{sec:appendix-framework-gf}.

Along with our lower bounds of section \ref{sec:lowerbounds}, this upper bound
lends evidence to our claim that variability captures the difficulty of
communicating changes in $f$ that are due to the non-monotonicity of the input
stream. A bolder claim is that variability is also useful in capturing the
difficulty of the distributed computation of a general function that is due to
the non-monotonicity of the input stream, but the extent to which that claim is
true has yet to be determined.

\section*{Acknowledgments}

Research supported in part by NSF grants CCF-0916574; IIS-1065276; CCF-1016540;
CNS-1118126; CNS-1136174; US-Israel BSF grant 2008411, OKAWA Foundation Research
Award, IBM Faculty Research Award, Xerox Faculty Research Award, B. John Garrick
Foundation Award, Teradata Research Award, and Lockheed-Martin Corporation
Research Award. This material is also based upon work supported by the Defense
Advanced Research Projects Agency through the U.S. Office of Naval Research
under Contract N00014-11-1-0392. The views expressed are those of the author and
do not reflect the official policy or position of the Department of Defense or
the U.S. Government.

\bibliographystyle{plain}
\bibliography{freqcount}


\appendix

\section{Variability of nearly monotone $f(n)$, theorem \ref{thm:smallvar-cases-nearlymonotone}}
\label{sec:appendix-variability-nearlymonotone}

\begin{theorem}
  Let $f^-(n) = \sum_{t : f'(t) < 0} |f'(t)|$ and $f^+(n) = \sum_{t : f'(t) > 0}
  f'(t)$. If there is a monotone nondecreasing function $\beta(t) \ge 1$ and a
  constant $t_0$ such that for all $n \ge t_0$ we have $f^-(n) \le \beta(n)
  f(n)$, then the variability $\sum_{t=1}^n |f'(t)/f(t)|$ is $O(\beta(n)
  \log(\beta(n) f(n)))$.
\end{theorem}

\begin{proof}
  For $i = 1, \ldots, k$, define $t_i$ to be the earliest time $t$ such that
  $f^+(t_i) > 2 f^+(t_{i-1})$, where $k$ is the smallest index such that $t_k >
  n$. (If $k$ is undefined, define $k = n+1$.)

  The cost $\sum_{t = 1}^{t_0-1} |f'(t)/f(t)|$ is constant. We bound the cost
  $\sum_{t = t_0}^{n} |f'(t)|/f(t)$ as follows. We partition the interval $[t_0,
    t_k)$ into subintervals $[t_0,t_1), \ldots, [t_{k-1}, t_k)$ and sum over the
        times $t$ in each one. There are at most $1 + \log f^+(n)$ of these
        subintervals.
  \begin{align*}
    \sum_{t=t_0}^n \frac{|f'(t)|}{f(t)} \;&\le\;
    \sum_{i=1}^k \sum_{t=t_{i-1}}^{t_i-1} \frac{|f'(t)|}{f(t)} \;\le\;
    \sum_{i=1}^k \frac{1 + \beta(n)}{f^+(t_{i-1})} \sum_{t=t_{i-1}}^{t_i-1} |f'(t)| \\ &\le\;
    \sum_{i=1}^k (1 + \beta(n)) \frac{f^+(t_i \m 1) + f^-(t_i \m 1)}{f(t_{i-1})} \\ &\le\;
    4 (1 + \beta(n)) (1 + \log f^+(t_i \m 1)) \\ &\le\;
    4 (1 + \beta(n)) (1 + \log(2 (1+\beta(n)) f(n)))
  \end{align*}
  because the condition $f^-(t) \le \beta(t) f(t)$ implies $f(t) \ge f^+(t) / (1
  + \beta(t))$ and $f^-(t) \le f^+(t)$.
\end{proof}

\section{Variability of biased coin flips, theorem \ref{thm:smallvar-cases-biased}}
\label{sec:appendix-variability-biased}

\begin{theorem}
  If $f'(t)$ is a sequence of i.i.d. $\pm 1$ random variables with $P(f'(t)
  \!=\! 1) = (1+\mu)/2$ then $E(v(n)) = O(\frac{\log n}{\mu})$.
\end{theorem}

\begin{proof}
  We show that, with high probability, $f(t) \ge \mu t / 2$ for times $t \ge t_0
  = t_0(n)$ when $n$ is large enough with respect to $\mu$.

  We write $f(t) = -t + 2 Y_t$, where $Y_t = \sum_{s = 1}^t y_s$, and $y_s$ is a
  Bernoulli variable with mean $\frac{1 + \mu}{2}$. We have that $P(f(t) \!\le\!
  \mu t / 2) = P(Y_t \!\le\! \frac{2 + \mu}{4} t)$ and that $E(Y_t) = \frac{1 +
    \mu}{2} t$. Using a Chernoff bound, $P(Y_t \!\le\! \frac{2 + \mu}{4} t) \le
  \exp(-\mu t / 16)$. Let $A$ be the event $\exists t \!\ge\! t_0 \, (f(t) \le
  \mu t / 2)$. Then $P(A) \le \sum_{t = t_0}^n e^{-\mu t / 16}$ by the union
  bound. We can upper bound this sum by
  \begin{align*}
    \sum_{t = t_0}^n e^{-\mu t / 16} \;\le\;
    e^{-\mu t_0 / 16} + \int_{t_0}^n e^{-\mu t / 16} \, dt \;\le\;
    17 e^{-\mu t_0 / 16} / \mu
  \end{align*}
  Taking $t_0 = (16 / \mu) \ln(17 n / \mu)$ gives us $P(A) \le 1/n$. Thus
  \begin{align*}
    E\left(\sum_{t=1}^n \min\{1, |\frac{f'(t)}{f(t)}|\}\right) \;\le\;
    t_0 + \left(\frac{1}{n}\right) n + \left(1 - \frac{1}{n}\right) \sum_{t=t_0}^n \frac{2}{\mu t} \;=\;
    O\left(\frac{\log n}{\mu}\right)
  \end{align*}
  yielding the theorem.
\end{proof}

\section{Simulating large $|f'(n)|$, section \ref{sec:upperbounds}}
\label{sec:appendix-upperbound-singleitem}

We noted in section \ref{sec:upperbounds} that we can simulate $|f'(n)| > 1$
with $|f'(n)|$ arrivals of $\pm 1$ updates with $O(\log \max f'(n))$ overhead.
To simplify notation we define $1/f(n) = 1$ when $f(n) = 0$ and assume that
$f(n) \ge 0$ always.

\begin{theorem}
  For $f'(n) > 1$ we have $\sum_{t=1}^{f'(n)} \frac{1}{f(n \m 1)+t} \le
  \frac{f'(n)}{f(n)} (1 + H(f'(n)))$ and for $f'(n) < -1$ we have
  $\sum_{t=0}^{1-f'(n)} \frac{t}{f(n)+t} \le \frac{3 f'(n)}{f(n)}$, where $H(x)$
  is the $x$th harmonic number.
\end{theorem}
\begin{proof}
  For $f'(n) > 1$, we have
    $\sum_{t=1}^{f'(n)} \frac{1}{f(n \m 1)+t} \;=\;
    \frac{f'(n)}{f(n)} + \frac{1}{f(n)} \sum_{t=1}^{f'(n)} \frac{f'(n)-t}{f(n \m 1)+t} \;\le\;
    \frac{f'(n)}{f(n)} + \frac{f'(n)}{f(n)} \sum_{t=1}^{f'(n)} \frac{1}{t}$.

  If $f'(n) < -1$ and $f(n) \ge 1$, then
    $\sum_{t=0}^{1-f'(n)} \frac{1}{f(n)+t} \le
    \frac{1}{f(n)} + \ln\left(\frac{f(n \m 1)}{f(n)}\right) =
    \frac{1}{f(n)} + \ln\left(1 + \frac{|f'(n)|}{f(n)}\right) \le
    \frac{2 |f'(n)|}{f(n)}$,
  and if $f(n) = 0$, add another $|f'(n)|/f(n)$.
\end{proof}

\section{Tracing and distributed tracking, section \ref{sec:lowerbounds}}
\label{sec:appendix-lowerbounds-lem0}

\begin{lemma}
  \label{lem:lowerbounds-lem0}
  Fix some $\ep$. Suppose that the tracing problem has an
  $\Omega(L_{\ep}(n))$-bit space deterministic lower bound. Also suppose that
  there is a deterministic algorithm $A$ for the distributed tracking problem
  that uses $\Omega(C_{\ep}(n))$ bits of communication and $\Omega(S_{\ep}(n))$
  bits of space at the site and coordinator combined. Then we must have $C + S =
  \Omega(L)$.

  Further, if we replace ``deterministic'' with ``randomized'' in the preceding
  paragraph, the claim still holds.
\end{lemma}

\begin{proof}
  Suppose instead that for all constants $c < 1$ and all $n_0$ there is an $n >
  n_0$ such that $C(n) + S(n) < c L(n)$. Then we can write an algorithm $B$ for
  the tracing problem that uses $L'(n) < c L(n)$ bits of space: simulate $A$,
  recording all communication, and on a query $t$, play back the communication
  that occurred through time $t$.

  At no point did we use the fact that $A$ guarantees $P(|f(t) \m \fhat(t)| \le
  \ep f(t)) = 1$, so the claim still holds if we change the correctness
  requirement to $P \ge 2/3$.
\end{proof}

\section{Deterministic lower bound, theorem \ref{thm:lowerbounds-det}}
\label{sec:appendix-lowerbounds-det}

\begin{theorem}
  Let $\ep = 1 / m$ for some integer $m \ge 2$, let $n \ge 2 m$, let $c < 1$
  constant, and let $r \le n^c$ and even. If a deterministic summary $S(f)$
  guarantees, even only for sequences for which $v(n) = \frac{6 m + 9}{2 m + 6}
  \ep r$, that $|f(t) - \fhat(t)| \le \ep f(t)$ for all $t \le n$, then that
  summary must use $\Omega(\frac{\log n}{\ep} v(n))$ bits of space.
\end{theorem}

\begin{proof}
  We construct a family of input sequences of length $n$ and variability
  $\frac{6 m + 9}{2 m + 6} \ep r$. Choose sets of $r$ different indices $1
  \ldots n$ so that there are $\text{choose}(n,r)$ such sets.

  For each set $S$ we define an input sequence $f_S$. We define $f_S(0) = m$ and
  the rest of $f_S$ recursively: $f_S(t) = f_S(t \m 1)$ if $t$ is not in $S$,
  and $f_S(t) = (2 m + 3) - f_S(t \m 1)$ if $t$ is in $S$. (That is, switch
  between $m$ and $m+3$.)

  If $A$ and $B$ are two different sets, then $f_A \ne f_B$: let $i$ be the
  smallest index that is in one and not the other; say $i$ is in $A$. Then
  $f_A(1 \ldots (i \m 1)) = f_B(1 \ldots (i \m 1))$, but $f_A(i) \ne f_A(i \m 1)
  = f_B(i \m 1) = f_B(i)$.

  The variability of any $f_S$ is $\frac{6 m + 9}{2 m + 6} \ep r$: There are
  $r/2$ changes from $m$ to $m+3$ and another $r/2$ from $m+3$ to $m$. When we
  switch from $m$ to $m+3$, we get $|f'(t)/f(t)| = 3/(m+3)$, and when we switch
  from $m+3$ to $m$, we get $|f'(t)/f(t)| = 3/m$. Thus $\sum_t
  |\frac{f'(t)}{f(t)}| = \frac{r}{2} \frac{6 m + 9}{m (m+3)} = \frac{6 m + 9}{2
    m + 6} \ep r$.

  There are $\text{choose}(n,r) \ge (n/r)^r$ input sequences in our family, so
  to distinguish between any two input sequences we need at least $r \log(n/r) =
  \Omega(r \log n)$ bits. Any summary that can determine for each $t$ the value
  $f(t)$ to within $\pm \ep f(t)$, must also distinguish between $f(t) = m$ and
  $f(t) = m+3$, since there is no value within $\ep m$ of $m$ and also within
  $\ep (m+3)$ of $m+3$. Since this summary must distinguish between $f(t) = m$
  and $f(t) = m+3$ for all $t$, it must distinguish between any two input
  sequences in the family, and therefore needs $\Omega(r \log n)$ bits.
\end{proof}

\section{Randomized lower bound, lemma \ref{lem:lowerbounds-rand-lem1}}
\label{sec:appendix-lowerbounds-rand-lem1}

\begin{lemma}
  Let $\FF$ be a family of sequences of length $n$ and variabilities $\le v$
  such that no two sequences in $\FF$ match. If a summary $S(f)$ guarantees for
  all $f$ in $\FF$ that $P(|f(t) - \fhat(t)| \le \ep f(t)) \;\ge\; 99/100$ for
  all $t \le n$, then that summary must use $\Omega(\log |\FF|)$ bits of space.
\end{lemma}

\begin{proof}
  Let $S(f)$ be the summary for a sequence $f$, and sample $\fhat(1) \ldots
  \fhat(n)$ once each using $S(f)$ to get $\fhat$. Let $A$ be the event that
  $|\{t \::\: |f(t) \m \fhat(t)| \le \ep f(t)\}| \;\ge\; \frac{90}{100} n$. By
  Markov's inequality and the guarantee in the premise, we must have $P(A) \ge
  9/10$.

  Let $\omega$ define the random bits used in constructing $S(f)$ and in
  sampling $\fhat$. For any choice $\omega$ in $A$ we have that $\fhat$ overlaps
  with $f$ in at least $\frac{9}{10} n$ positions, which means that $\fhat$
  overlaps with any other $g \in \FF$ in at most $\frac{7}{10} n$ positions: at
  most the $\frac{6}{10} n$ in which $f$ and $g$ could overlap, plus the
  $\frac{1}{10} n$ in which $\fhat$ and $f$ might not overlap.

  Define $F \subseteq \FF$ to be the sequences $g$ that overlap with $\fhat$ in
  at least $\frac{9}{10} n$ positions. This means that when $\omega \in A$ we
  have $|F| = 1$, and therefore with probability at least $9/10$ we can identify
  which sequence $f$ had been used to construct $S(f)$.

  We now prove our claim by reducing the $\text{Index}_N$ problem to the problem
  of tracing the history of a sequence $f$. The following statement of
  $\text{Index}_N$ is roughly as in Kushilevitz and Nisan \cite{KN1997}. There
  are two parties, Alice and Bob. Alice has an input string $x$ of length $N =
  \log_2 |\FF|$ and Bob has an input string $i$ of length $\log_2 N$ that is
  interpreted as an index into $x$. Alice sends a message to Bob, and then Bob
  must output $x_i$ correctly with probability at least $9/10$.

  Consider the following algorithm for solving $\text{Index}_N$. Alice
  deterministically generates a family $\FF$ of sequences of length $n$ and
  variabilities $\le v$ such that no two match, by iterating over all possible
  sequences and choosing each next one that doesn't match any already chosen.
  Her $\log_2 |\FF|$ bits of input $x$ index a sequence $f$ in $\FF$. Alice
  computes a summary $S(f)$ and sends it to Bob. After receiving $S(f)$, Bob
  computes $\fhat(t)$ for every $t = 1 \ldots n$, to get a sequence $\fhat$. He
  then generates $\FF$ himself and creates a set $F$ of all sequences in $\FF$
  that overlap with $\fhat$ in at least $\frac{9}{10} n$ positions. If $F =
  \{f\}$, which it is with probability at least $9/10$, then Bob can infer every
  bit of $x$.

  Since the $\text{Index}_N$ problem is known to have a one-way communication
  complexity of $\Omega(N)$, it must be that $|S(f)| = \Omega(\log |\FF|)$.
\end{proof}

\section{Randomized lower bound, lemma \ref{lem:lowerbounds-rand-lem2}}
\label{sec:appendix-lowerbounds-rand-lem2}

\begin{lemma}
  For all $\ep \le 1/2$, $v \ge 32400 \ep \ln C$, and $n > 3 v / \ep$, there is
  a family $\FF$ of size $e^{\Omega(v/\ep)}$ of sequences of size $n$ such that:
  \begin{enumerate*}
  \item no two sequences match, and
  \item every sequence has variability at most $v$.
  \end{enumerate*}
\end{lemma}

\begin{proof}
  We construct $\FF$ so that each of the two items holds (separately) with
  probability at least $4/5$. Let $m = 1/\ep$. To construct one sequence in
  $\FF$, first define $f(0) = m$ with probability $1/2$, else $f(0) = m \!+\!
  3$. Then, for $t = 1 \ldots n$: define $f(t) = (2m \!+\! 3) - f(t \m 1)$ with
  probability $p = v/6 \ep n$, else $f(t) = f(t \m 1)$. That is, switch from $m$
  to $m \!+\! 3$ (or vice-versa) with probability $p = v/6 \ep n$.

  We first prove that the probability is at most $1/5$ that any two sequences
  $f$ and $g$ match. We have that $P(f(0) \eq g(0)) = 1/2$. If at any point in
  time we have $f(t) = g(t)$, then $P(f(t \!+\! 1) \eq g(t \!+\! 1)) = \alpha =
  1 - 2 p (1 \m p)$ and $P(f(t \!+\! 1) \!\ne\! g(t \!+\! 1)) = 1 - \alpha = 2 p
  (1 \m p)$. Similarly, if $f(t) \ne g(t)$, then $P(f(t \!+\! 1) \eq g(t \!+\!
  1)) = 1 - \alpha$ and $P(f(t \!+\! 1) \!\ne\! g(t \!+\! 1)) = \alpha$.

  The overlap between $f$ and $g$ is the number of times $t$ that $f(t) = g(t)$.
  We model this situation with a Markov chain $M$ with two states, $c$ for
  ``same'' (that is, $f = g$) and $d$ for ``different'' ($f \ne g$). Let $s_t$
  be the state after $t$ steps, and let $p_t = (p_t(c), p_t(d))$ be the
  probabilities that $M$ is in state $c$ and $d$ after step $t$. The stationary
  distribution $\pi = (1/2, 1/2)$, which also happens to be our initial
  distribution. We can model the overlap between $f$ and $g$ by defining a
  function $y(s_t) = 1$ if $s_t = c$ and $y(s_t) = 0$ otherwise; then $Y =
  \sum_{t = 1}^n y(s_t)$ is the overlap between $f$ and $g$. The expected value
  $E(y(\pi))$ of $y$ evaluated on $\pi$ is $1/2$.

  The $(1/8)$-mixing time $T$ is defined as the smallest time $T$ such that
  $\frac{1}{2}||M^t r_0 - \pi||_1 \le 1/8$ over all initial distributions $r_0$.
  Let $r_0$ be any initial distribution and $r_t = M^t r_0$. If we define
  $\Delta_t = r_t(c) - \pi(c)$, then $\Delta_t = (2 \alpha \m 1)^t \Delta_0$. We
  can similarly bound $|r_t(d) - \pi(d)|$, so we can bound
  \begin{align*}
    T
    \;&\le\; \frac{\ln(8)}{\ln(1/(2 \alpha \m 1))}
    \;\le\; \frac{3}{(1 - (2 \alpha \m 1))}
    \;\le\; \frac{3}{2 p (1 \m p)}
    \;\le\; \frac{3}{2 p}
    \;=\; \frac{9 \ep n}{v}
  \end{align*}
  since $1 - p \ge 1/2$ and since $1/\ln(1/x) \le 1/(1 \m x)$ for $x$ in
  $(0,1)$. With this information we can now apply a sledgehammer of a result by
  Chung, Lam, Liu, and Mitzenmacher \cite{CLLM2012}.
  Our fact \ref{fac:lowerbounds-rand-lem-cllm} is their theorem 3.1, specialized
  a bit to our situation:
  \begin{fact}
    \label{fac:lowerbounds-rand-lem-cllm}
    Let $M$ be an ergodic Markov chain with state space $S$. Let $T$ be its
    $(1/8)$-mixing time. Let $(s_1, \ldots, s_n)$ denote an $n$-step random walk
    on $M$ starting from its stationary distribution $\pi$. Let $y$ be a weight
    function such that $E(y(\pi)) = \mu$. Define the total weight of the walk by
    $Y = \sum_{t = 1}^n y(s_t)$. Then there exists some universal constant $C$
    such that $P(Y \ge (1 + \delta) \mu n) \;\le\; C \exp(-\delta^2 \mu n / 72
    T)$ when $0 < \delta < 1$.
  \end{fact}
  Specifically, this means that $P(Y \ge \frac{6}{10} n) \;\le\; C \exp(-v / (25
  \cdot 72 \cdot 9 \cdot \ep))$. Since $v$ is large enough, we can also write $P
  \le \exp(-v / 32400 \ep)$. If $|\FF| = \frac{1}{5} \exp(v / (2 \cdot 32400
  \ep))$, then by the union bound, with probability at least $4/5$, no pair of
  sequences $f, g$ matches.

  We also must prove that there are enough sequences with variability at most
  $v$. The change in variability due to a single switch from $m$ to $m \!+\! 3$
  (or vice-versa) is at most $3/m = 3 \ep$. For any sequence $f$, let $U_t = 1$
  if $f$ switched at time $t$, else $U_t = 0$. The expected number of switches
  is $v / 6 \ep$; using a standard Chernoff bound, $P(\sum_t U_t \ge 2 v / 6
  \ep) \;\le\; \exp(-v / 18 \ep) \;\le\; 1/10$. Suppose we sample $N$ sequences
  and $B$ of them have more than $2 v / 6 \ep$ switches. In expectation there
  are at most $E(B) \le \frac{1}{10} N$ that have too many switches. By Markov's
  inequality, $P(B \ge N/2) \;\le\; 1/5$, so we can toss out the $\le N/2$ bad
  sequences. This gives us a final size of $\FF$ of $\frac{1}{10} \exp(v / (2
  \cdot 32400 \ep))$.
\end{proof}

\section{Tracking item frequencies, section \ref{sec:framework-itemfreq}}
\label{sec:appendix-framework-itemfreq}

\paragraph{Problem definition}
The problem of tracking item frequencies is only slightly different than the
counting problem we've considered so far. In this problem there is a universe
$U$ of items and we maintain a dataset $D(t)$ that changes over time. At each
new timestep $n$, either some item $\ell$ from $U$ is added to $D$, or some item
$\ell$ from $D$ is removed. This update is told to a single site $i$; that is,
site $i(n)$ receives an update $f_{\ell}'(n) = \pm 1$.

The frequency $f_{\ell}(t)$ of item $\ell$ at time $t$ is the number of copies
of $\ell$ that appear in $D(t)$. The first frequency moment $F_1(t)$ at time $t$
is the total number of items $|D(t)|$. The problem is to maintain estimates
$\fhat_{\ell}(n)$ at the coordinator so that for all times $n$ and all items
$\ell$ we have that $P(|f_{\ell}(n) \m \fhat_{\ell}(n)| \le \ep F_1(n))$ is
large.

Since in this problem we are tracking each item frequency to $\ep F_1(n)$, we
use $F_1$-variability instead, defining $v'(t) = \min\{1, 1/F_1(t)\}$.

\subsubsection{Item frequencies with low communication}
We first partition time into blocks as in section \ref{sec:upperbounds-block},
using $f = F_1$. That is, at the end of each block we know the values $n$ and
$F_1(n)$ deterministically, and also that either $r = 0$ holds or that
$F_1(n_j)$ is within a factor of two of $F_1(n_{j-1})$.

For tracking during blocks we modify the deterministic algorithm so that each
site $i$ holds counters $d_{i\ell}$ and $\delta_{i\ell}$ for every item $\ell$.
It also holds counters $f_{i\ell}$ of the total number of copies of $\ell$ seen
at site $i$ across all blocks.

At the end of each block, each site $i$ reports all $f_{i\ell} \ge \ep 2^r / 3$
(using the new value of $r$). If site $i$ reports counter $f_{i\ell}$ then it
starts the next block with $d_{i\ell} = \delta_{i\ell} = 0$; otherwise,
$d_{i\ell}$ is updated to $d_{i\ell} + \delta_{i\ell}$ and then $\delta_{i\ell}$
is reset to zero. Within a block $r \ge 1$, the \cnd{} is true when
$\delta_{i\ell} \ge \ep 2^r / 3$.

The coordinator maintains estimates $\fhat_{i\ell}$ of $f_{i\ell}$ for each site
$i$ and item $\ell$. Upon receiving an update $\delta_{i\ell}$ during a block
the coordinator updates its estimate $\fhat_{i\ell} = \fhat_{i\ell} +
\delta_{i\ell}$.

\paragraph{Estimation error}
The total error in the estimate $\fhat_{i\ell}(n)$ at any time $n$ is the error
due to $d_{i\ell}$ plus the error due to $\delta_{i\ell}$. In both cases these
quantities are bounded by $\ep 2^r / 3 \le \ep F_1(n) / 3$.

\paragraph{Communication}
The total communication for a block is the total communicated within and at the
end of the block. Within a block, all $\delta_{i\ell}$ start at zero, and there
are at most $2^r k$ updates, so the total number of messages sent is $3 k /
\ep$. At the end of a block, $f_{i\ell} \ge \ep 2^r / 3$ is true for at most $12
k / \ep$ counters $f_{i\ell}$. Therefore the total number of messages
$O(\frac{k}{\ep} v(n))$.

\subsubsection{Item frequencies in small space+communication}
The algorithm so far uses $|U|$ counters per site, which is prohibitive in terms
of space. In \cite{CM2005} Cormode and Muthukrishnan show that in order to track
over a non-distributed update stream each $f_{\ell}(n)$ so that for all $\ell$
and all times $n$ we have $P(|f_{\ell}(n) \m \fhat_{\ell}(n)| \le \ep F_1(n) /
3) \:\ge\: 8/9$, it suffices to randomly partition each item in $U$ into one of
$27 / \ep$ classes using a pairwise-independent hash function $h$, and to
estimate $f_{\ell}(n)$ as $f_{h(\ell)}(n)$. The $27 / \ep$ counters and the hash
function $h$ together form their Count-Min Sketch \cite{CM2005}.

Similarly, in \cite{GM2006} \cite{GM2007} Ganguly and Majumder adapt a data
structure of Gasieniec and Muthukrishnan \cite{Muthukrishnan2005}, which they
call the CR-precis, to deterministically track each $f_{\ell}(n)$ to $\ep F_1(n)
/ 3$ error. This data structure uses $\frac{3}{\ep}$ rows of $\frac{6 \log
  |U|}{\ep \log 1/\ep}$ counters, and estimates $f_{\ell}(n)$ as the average
over rows $r$ of $f_{h(r,\ell)}(n)$. (Ganguly and Majumder actually take the
minimum over the rows $r$, but the average works too and yields a linear
sketch.)

In either case, we can first reduce our set of items $\ell$ to a small number of
counters $c$, and instead of tracking $f_{i\ell}$ we track $f_{ic}$ for each
counter $c$. The coordinator can then linearly combine its estimates
$\fhat_{ic}$ to obtain estimates $\fhat_{i\ell}$ for each item $\ell$. This
introduces another $\ep F_1(n) / 3$ error, yielding algorithms that guarantee
\begin{itemize*}
\item $P(|f_{i\ell}(n) \!-\! \fhat_{i\ell}(n)| \le \ep F_1(n)) \:=\: 1$ in
  $O(\frac{k \log |U|}{\ep^2 \log 1/\ep} v(n) \log n)$ bits of space +
  communication, and
\item $P(|f_{i\ell}(n) \!-\! \fhat_{i\ell}(n)| \le \ep F_1(n)) \:\ge\: 8/9$ in
  $O(k \log |U| + \frac{k}{\ep} v(n) \log n)$ bits of space + communication.
\end{itemize*}

\subsubsection{Remarks}
We obtain a randomized communication bound of $O(\frac{k}{\ep} v(n))$ messages,
but it might be possible to do better. In \cite{HYZ2012} Huang et. al. both
develop a randomized counting algorithm ($O(\frac{\sqrt{k}}{\ep} \log n)$
messages) and also extend it to the problem of tracking item frequencies to get
the same communication bound. Unfortunately, their algorithm appears to require
the total variance in their estimate at any time $t < n$ to be bounded by a
constant factor of the variance at time $n$. This is only guaranteed to be true
when item deletions are not permitted (and $F_1$ grows monotonically).
We avoid this problem in section \ref{sec:upperbounds-rand} for tracking $f =
F_1$ by deterministically updating $F_1$ at the end of each block. For this
problem, though, deterministically updating all of the large $\fhat_{i\ell}$ at
the end of each block could incur $O(1/\ep)$ messages.
Whether it is also possible to probabilistically track item frequencies over
general update streams in $O(\frac{\sqrt{k}}{\ep} v(n))$ messages remains open.

\section{Aggregate functions with one site, section \ref{sec:framework-gf}}
\label{sec:appendix-framework-gf}

The single-site algorithm of section \ref{sec:framework-gf} is: whenever $|f -
\fhat| > \ep f$, send $f$ to the coordinator.

\begin{proof}
  If $f(n) = 0$ then $v'(n) = 1$. Also, if $f(n)$ changes sign from $f(n \m 1)$,
  then $v'(n) = 1$. So consider intervals over which $f(n)$ is nonzero and
  doesn't change sign. Over such an interval, let $\Phi(n) = |\frac{f(n) -
    \fhat(n)}{f(n)}|$. If at time $n$ we update $\fhat$ then $\Phi(n) = 0$.
  Otherwise,
  \begin{align*}
    \Phi(n)
    \;&=\; \frac{|f(n \m 1) - \fhat(n \m 1) + f'(n)|}{|f(n)|}
    \;\le\; \frac{|f(n \m 1) - \fhat(n \m 1)|}{|f(n)|} \:+\:
            \frac{|f'(n)|}{|f(n)|} \\
    \;&=\; \frac{|f(n \m 1)|}{|f(n)|} \Phi(n \m 1) \:+\:
           \frac{|f'(n)|}{|f(n)|}
    \;\le\; \frac{|f(n)| + |f'(n)|}{|f(n)|} \Phi(n \m 1) \:+\:
            \frac{|f'(n)|}{|f(n)|} \\
    \;&\le\; \Phi(n \m 1) \:+\:
            \frac{(1 \!+\! \Phi(n \m 1)) |f'(n)|}{|f(n)|}
  \end{align*}
  Since $\Phi(n) \le \ep$ we have $|\Phi'(n)| \le (1 \!+\! \ep)
  |\frac{f'(n)}{f(n)}|$. We only send a message each time that $\Phi$ would be
  more than $\ep$, so the total number of messages sent is at most the total
  increase in $\Phi$, which is $\sum_{t=1}^n \min \{1, |\frac{f'(t)}{f(t)}|\}$.
\end{proof}

\end{document}